\documentclass[letterpaper, 10 pt, conference]{ieeeconf}  

\IEEEoverridecommandlockouts                              
\usepackage{graphicx} 
\usepackage{amsmath} 
\usepackage{amsfonts,amssymb} 
\usepackage{algorithm}
\usepackage{algorithmic} 
\usepackage{booktabs} 
\usepackage{makecell} 
\newtheorem{lemma}{Lemma} 

\title{\LARGE \bf
Nonlinear Kalman Filtering based on Self-Attention Mechanism 
and Lattice Trajectory Piecewise Linear Approximation
}

\author{Jiaming Wang, Xinyu Geng and Jun Xu
\thanks{*This work was supported in part by the National Natural Science Foundation of China under Grant 62173113, and 
in part by the Science and Technology Innovation Committee of Shenzhen Municipality under Grant GXWD20231129101652001, 
and in part by Natural Science Foundation of Guangdong Province of China under Grant 2022A1515011584.}
\thanks{Jiaming Wang, Xinyu Geng and Jun Xu are with School of Mechanical Engineering and Automation,
Harbin Institute of Technology, Shenzhen, 518055, China
        {\tt\small 21s153144@stu.hit.edu.cn, 22s153095@stu.hit.edu.cn, xujunqgy@hit.edu.cn}}%
}

\begin{document}

\maketitle
\thispagestyle{empty}
\pagestyle{empty}

\begin{abstract}

The traditional Kalman filter (KF) is widely applied in control systems, 
but it relies heavily on the accuracy of the system model and noise parameters, 
leading to potential performance degradation when facing inaccuracies.
To address this issue, 
introducing neural networks into the KF framework offers a data-driven solution 
to compensate for these inaccuracies, improving the filter's performance while maintaining interpretability. 
Nevertheless, existing studies mostly employ recurrent neural network (RNN), 
which fails to fully capture the 
dependencies among state sequences and lead to an unstable training process. 
In this paper, we propose a novel 
Kalman filtering algorithm named the attention Kalman filter (AtKF), 
which incorporates a self-attention 
network to capture the dependencies among state sequences. To address the instability in the recursive training 
process, a parallel pre-training strategy is devised. Specifically, this strategy involves piecewise linearizing the 
system via lattice trajectory piecewise linear (LTPWL) expression, and generating pre-training data through a batch 
estimation algorithm, which exploits the self-attention mechanism's parallel processing ability.
Experimental results on a two-dimensional nonlinear system demonstrate that AtKF outperforms other filters under noise disturbances and model mismatches.
       
\end{abstract}

\section{INTRODUCTION}

In the field of modern control theory, 
the Kalman filter (KF) \cite{kalman1960} and its variant, the extended Kalman filter 
(EKF) \cite{maybeck1982stochastic}, are fundamental tools for state estimation in control system design.
However, the performance of these model-based filters depends significantly on the accuracy of the system model and noise 
parameters. 
Inaccurate settings can lead to a notable decline in KF's performance.

To address this challenge, many studies improved KF by integrating data-driven approaches, 
which are mainly categorized into external combination and internal embedding \cite{bai2023state}. 
External combination approaches employ neural networks for either identifying system parameters or enhancing fusion 
filtering with KF, where the neural network works independently of KF. 
Gao et al. \cite{gao2020rl} proposed an adaptive KF that uses reinforcement learning to estimate process noise covariance 
dynamically, thus improving navigation accuracy and robustness. 
Tian et al. \cite{tian2021state} devised a battery state estimation method that merges the outputs of a deep neural network 
with the ampere-hour counting method through a linear KF, yielding faster and more precise estimations.
Internal embedding strategies integrate neural networks within the KF framework and replace certain parts of the traditional KF.
Jung et al. \cite{jung2020mnemonic} introduced a memorized KF that uses long short-term memory (LSTM) networks 
\cite{LongShortTermMemory} to learn transition probability density functions.
This approach effectively surpasses the Markovian and linearity constraints inherent in traditional KF.
KalmanNet \cite{revach2022kalmannet} combined KF with gated recurrent units 
(GRU) \cite{chung2014empirical} to estimate the 
Kalman gain, showing improved filtering performance in model mismatched and nonlinear systems.
Directly embedding neural networks into the KF framework represents a novel and promising research direction.

However, most current approaches employ LSTM or GRU to learn from time series data. 
These recurrent neural networks (RNN) perform poorly in comprehensively capturing the dependencies in time series data.
Additionally, their recursive training processes suffer from instability and inefficiency.

Inspired by KalmanNet \cite{revach2022kalmannet}, we introduce a novel technique that incorporates the self-attention mechanism 
from Transformer \cite{vaswani2017attention} into the Kalman filtering.
By analyzing state sequences over historical time windows, our method aims to capture dependencies among state sequences more 
effectively, thereby enhancing estimation accuracy and robustness. 
However, due to KF's recursive structure, directly applying the attention mechanism within KF leads to an inherently recursive 
training process, which is incapable of addressing the issues of instability and inefficiency. To solve this, we design a 
pre-training method that constructs all pre-training data through batch estimation. 
It estimates the system states over a period in one go, 
thereby avoiding the recursive process. 
This approach sets up better starting points for the attention 
network, enabling it to replicate the benefits of extensive training through a minimal number of iterations.  

Nevertheless, for batch estimation of nonlinear systems, 
it is necessary to perform linearization first, 
for which the lattice trajectory piecewise linear (LTPWL) expression 
offers an analytical and compact solution. The lattice piecewise linear (PWL) expression is named for its algebraic properties 
of performing max and min operations on affine functions \cite{tarela1999region}. 
Tarela et al. \cite{tarela1999region} summarized several representation methods of lattice PWL functions from \cite{lin1994explicit} 
\cite{tarela1990representation}. 
Ovchinnikov \cite{ovchinnikov2002max} provided proof that lattice PWL can represent any PWL function, 
and Xu et al. \cite{xu2016irredundant} introduced methods for removing redundant terms and literals in lattice PWL. 
Wang et al. \cite{wang2021lattice} proposed a LTPWL method for approximating nonlinear systems with lattice PWL.
Here, we use the LTPWL to perform piecewise linearization of the nonlinear system, 
then generate pre-training data through a batch estimation algorithm for non-nested training of the network.

The main contributions of this paper can be summarized as follows:
\begin{itemize}
    \item A Kalman filtering algorithm embedded with a simplified attention mechanism is proposed, 
    which better captures the dependencies among state sequences,  thereby improving the accuracy and robustness of 
    state estimation.
    \item A pre-training method based on the LTPWL and batch estimation algorithm is designed, 
    addressing the instability and inefficiency of the recursive training process, 
    while fully leveraging the parallel processing capabilities of the self-attention network.
\end{itemize}

The paper is structured as follows: Section 2 introduces the self-attention mechanism and LTPWL expression.
Section 3 details the structure of AtKF and the pre-training method. Section 4 evaluates our approach through experiments 
on a two-dimensional nonlinear system, and Section 5 concludes the paper.


\section{PRELIMINARIES}

\subsection{Self-attention Mechanism}

\begin{figure} 
    \centering
    \includegraphics[width=0.35\textwidth]{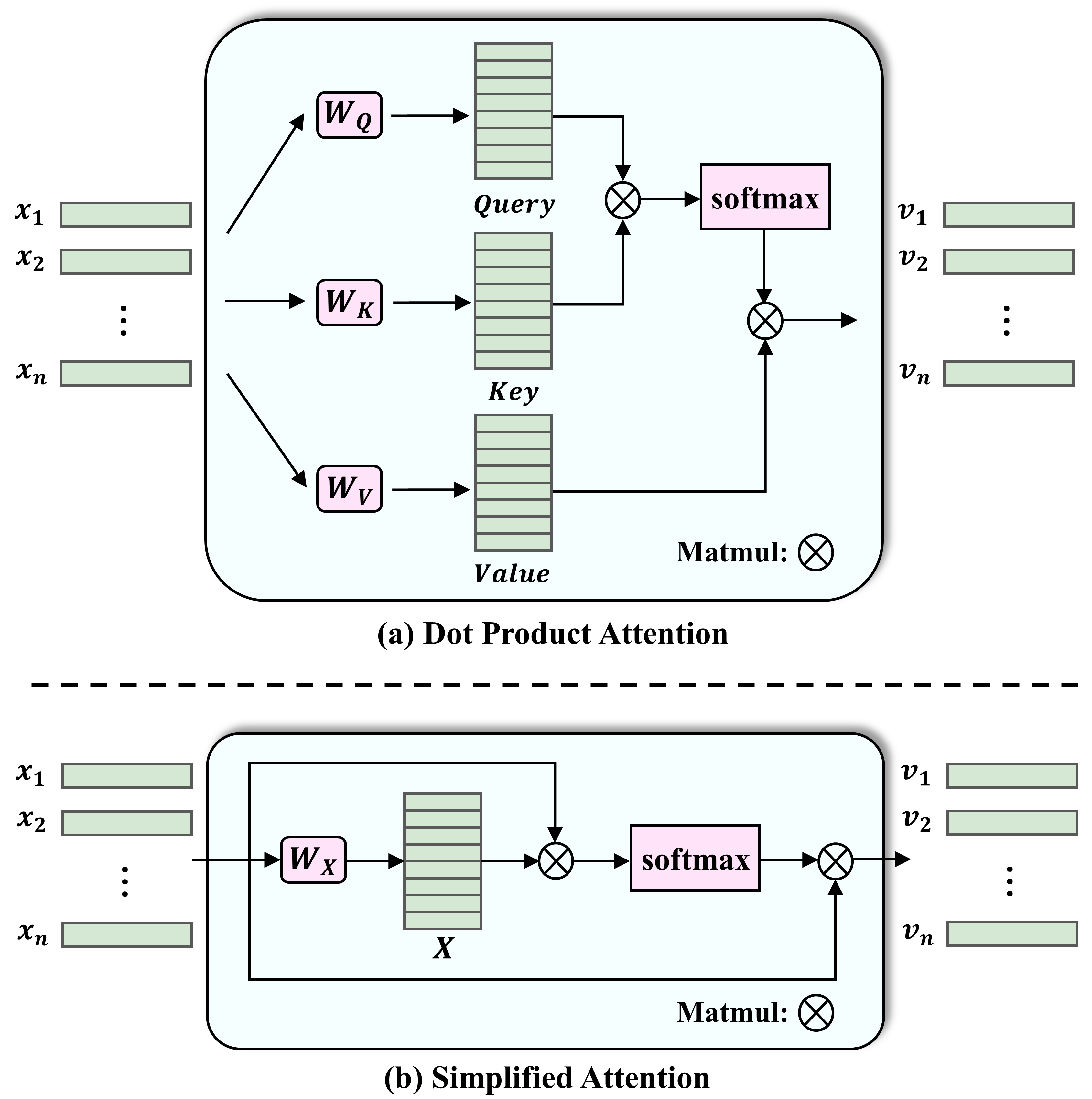} 
    \caption{(a) Self-attention Mechanism, (b) Simplified Attention Network.}
    \label{selfAttention}
\end{figure}

This section introduces the operation of the attention mechanism and our proposed simplified attention network. The self-attention 
mechanism is shown in Fig. \ref{selfAttention}(a). 
It transforms the input sequence ${x_{1},x_{2}, \ldots, x_{n}}$ into query, key, and value matrices through 
linear mappings. For each sequence element $x_i$, it calculates the dot products with all sequence elements,
forming an attention distribution through softmax that indicates the elements' dependencies.  
This attention distribution is then multiplied with value to produce the output sequence $\{v_{1}, v_{2}, \ldots, v_{n}\}$, 
where each element $v_i$ integrates information from the entire sequence. 
This ensures each processed element reflects the influence of every other element, overcoming the limitations of distance.

A simplified version of this mechanism, as depicted in Fig. \ref{selfAttention}(b), 
streamlines the process by using a single matrix $X$ for multiplying both the attention distribution and the output sequence.
This is feasible because $d_{embed}=d_{model}$, and the dimension of $X$ matches that of the input sequence matrix, 
eliminating the need for separate and extra linear mappings. 
Given the small amount of sequence data and the simple distribution of features in the Kalman filtering process, 
employing the full multi-head attention mechanism can increase training difficulty and lead to overfitting. 
By reducing the number of parameters, this simplified approach boosts efficiency without sacrificing the model's 
ability to capture crucial dependencies.

\subsection{Lattice Trajectory Piecewise Linear Expression}

\begin{figure} 
   \centering
   \includegraphics[width=0.3\textwidth]{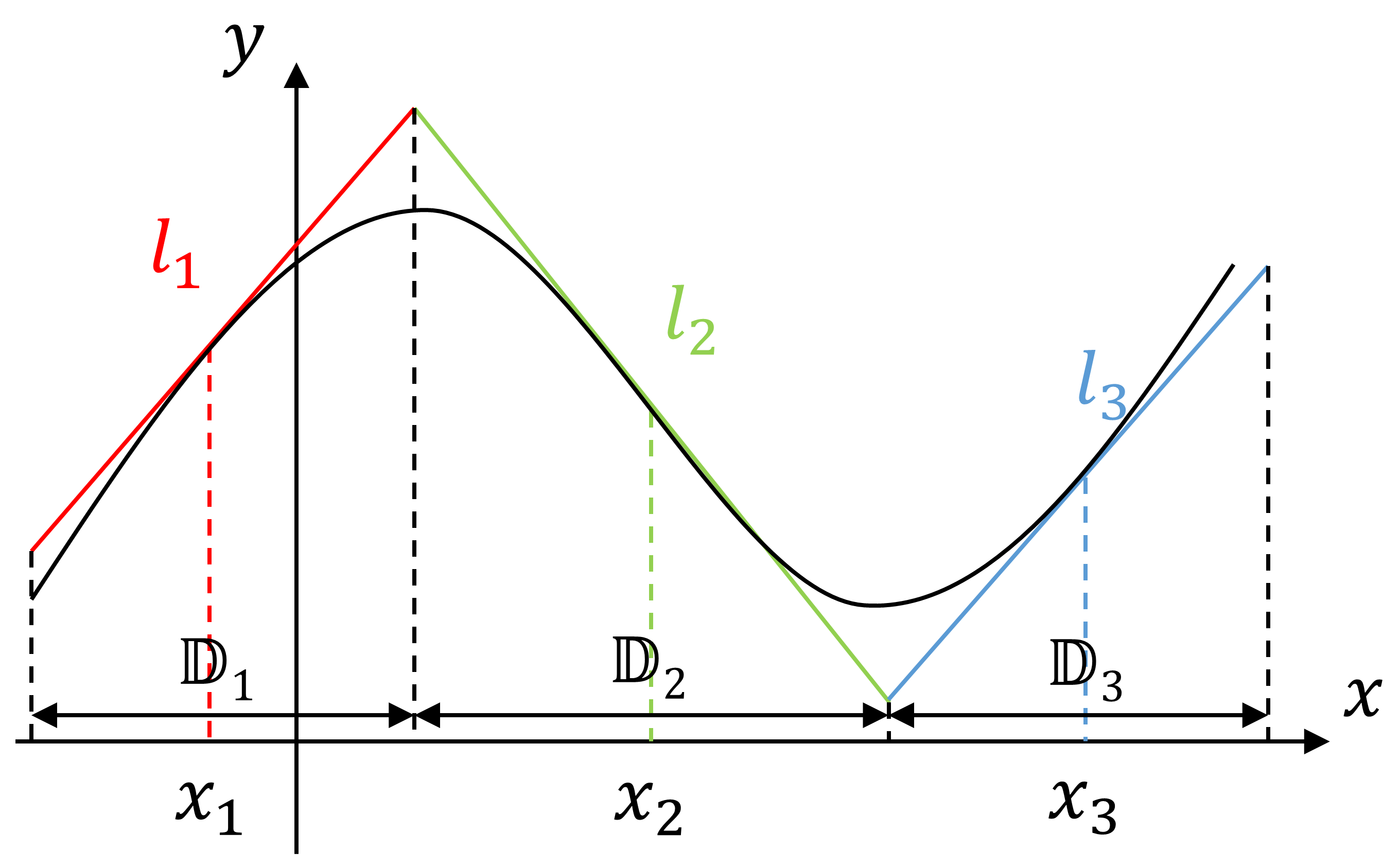}
   \caption{An example of lattice trajectory piecewise linear expression.}
   \label{LTPWL_example}
\end{figure}

The LTPWL expression is an approximate method of constructing a lattice PWL expression for a nonlinear function.
It can simultaneously accomplish the linearization of the nonlinear system and the construction of 
lattice PWL expression.
Its main process involves selecting a set of linearization points along the system's state 
trajectory, constructing linear segments at each of these points, and finally using these segments to build a 
lattice PWL expression to approximate the original nonlinear system.

Taking the nonlinear function shown in Fig. \ref{LTPWL_example} as an example, by selecting
points \(x_{1}, x_{2},\) and \(x_{3}\) along the state trajectory (here, the $x$-axis), and creating corresponding 
linear segments \(l_{1}, l_{2},\) and \(l_{3}\) at each point, we represent the PWL function as shown 
in (\ref{latticePWL}),
\begin{equation} \label{latticePWL}
   f(x) = \mathop{\max}\{\mathop{\min}\{l_1,l_2\},\mathop{\min} \{l_1,l_3\}\},\forall x \in \mathbb{R},
\end{equation}
this method simplifies the approximation of nonlinear functions without defining specific interval ranges for 
$x$, making it a more compact solution compared with traditional piecewise form (\ref{conventionalPWL}),
\begin{equation} \label{conventionalPWL}
   f(x) =
   \begin{cases}
      l_1(x), \quad x \le x_{1}, \\
      l_2(x),\quad x_{1} \le x \le x_{2}, \\
      l_3(x), \quad x \ge x_{2}.
   \end{cases}
\end{equation}

The general form of the lattice PWL is shown in (\ref{latticeRepresent}),
\begin{equation} \label{latticeRepresent}
    f = \mathop{\max}\limits_{i=1,\dots,N} \left \{ \mathop{\min}_{j \in \tilde{I}_{\ge,i}}\{l_j\} \right \},
\end{equation}
where $N$ is the number of base regions $\mathbb{D}_i$, with each $\mathbb{D}_i$ being a polyhedron satisfying 
$\{x|l_i(x) = l_j(x), j \neq i\} \cap \mathbb{D}_i = \emptyset$. Furthermore, $\tilde{I}_{\ge,i}$ denotes the 
index set of affine functions in the base region $\mathbb{D}_i$ that are greater than or equal to the linear segment 
$l_i(x)$, i.e., $\tilde{I}_{\ge,i} = \{j|l_j(x) \geq l_i(x), \forall x \in \mathbb{D}_i\}$.

For the PWL function shown in Fig. \ref{LTPWL_example} with three basic regions, 
$N=3$. Taking the linear segment $l_1$ as an example, 
in the basic region $\mathbb{D}_1$, 
the affine functions that are greater than or equal to $l_1$ are $l_1$ and $l_2$. 
Then, a minimum operation is applied to $l_1$ and $l_2$ to construct a ``term" in the lattice PWL expression, 
resulting in $\mathop{\min}\{l_1,l_2\}$. 
Terms for linear segments $l_2$ and $l_3$ are constructed in a similar way, 
resulting in $\mathop{\min}\{l_1,l_2\}$ and $\mathop{\min}\{l_1,l_3\}$, respectively. 
Then, a maximum operation on these three terms yields the final lattice PWL expression (\ref{latticePWL}). 
When evaluating the lattice piecewise linear expression, 
it is only necessary to substitute the value of the variable $x$, 
without the need to consider the interval range of the variable as in the traditional piecewise form. 
Moreover, we can conveniently identify the linear segment $l_{i}$ 
that represents the current system dynamics through comparison operations. 
As we will see in Section \ref{preTrainSection}, 
this property of the lattice expression is particularly suited for batch estimation algorithms to generate pre-training data.


\section{Kalman Filtering Algorithm with Attention Mechanism}

This section offers an overview of the AtKF framework 
and the pre-training approach based on the LTPWL expression. 
It is structured into four parts: the system model, the overall architecture, the network structure, and the training 
methodology.

\subsection{System Model}
Considering the discrete-time nonlinear system given by (\ref{nonlinearSys}),
\begin{subequations} \label{nonlinearSys}
   \begin{align} 
           x_{k} &= f(x_{k-1}) + w_{k}, \label{stateEquation} \\
           y_{k} &= h(x_{k}) + v_{k}, \label{outputEquation} \\  
           w_{k} &\sim N(0, Q), v_{k} \sim N(0, R),
   \end{align}
\end{subequations}
where $f$ and $h$ represent the nonlinear state transition and observation functions, respectively, $x_k$ denotes 
the state vector at time step $k$, and $y_k$ represents the observation at $k$. $w_k$ and $v_k$ correspond to the 
process noise and observation noise at $k$, respectively, both assumed to be Gaussian white noise with their 
covariance matrices $Q$ and $R$.

\begin{figure*} 
    \centering
    \includegraphics[width=0.9\textwidth]{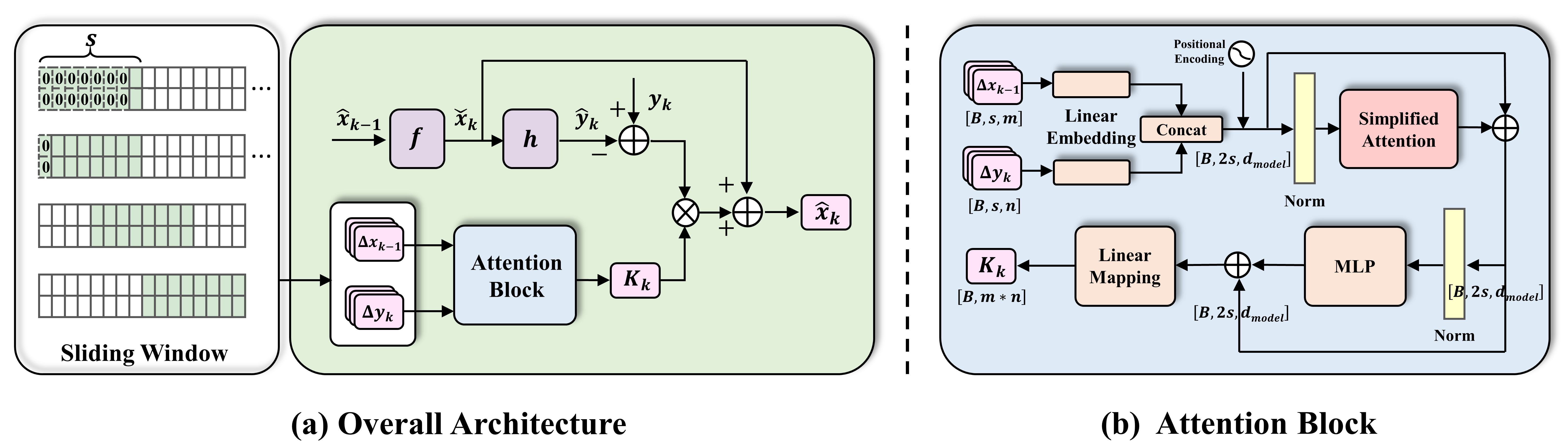}
    \caption{(a) overall architecture, (b) self-attention mechanism network.}
    \label{OverallArch}
\end{figure*}

\subsection{Overall Architecture}

As shown in Fig. \ref{OverallArch}(a), the framework aligns with the traditional Kalman filtering algorithm, 
using the system model for state recursion and output prediction at each time step. 
At any given moment, such as time step $k$, 
$\check{x}_{k}$ and $\hat{y}_{k}$ represent the prior estimate of the state $x_{k}$ and the predicted system output, respectively, 
while $\hat{x}_{k}$ is the posterior state estimate. 
A self-attention mechanism is incorporated to predict the Kalman gain $K_{k}$. 
The gain, acting as a fusion coefficient between model predictions and system observations, 
adjusts based on process $w_{k}$ and observation noise $v_{k}$. 
By leveraging the self-attention mechanism, the network captures sequential data dependencies more effectively, 
leading to enhanced estimation accuracy by fitting the system's noise characteristics.

The input to the self-attention mechanism network comprises two types of features. 
which can be expressed as in (\ref{inputFeature}) at time step $k$,
\begin{subequations} \label{inputFeature}
    \begin{align}
        \Delta x_{k-1} &= \hat{x}_{k-1} - \check{x}_{k-1} \label{feature1}, \\
        \Delta y_{k}   &= y_{k} - \hat{y}_{k} \label{feature2},
    \end{align}
\end{subequations}
$\Delta x_{k-1} \in \mathbb{R}^{m}$ represents the forward update difference, 
and $\Delta y_{k} \in \mathbb{R}^{n}$ represents the innovation. 
To fully utilize the sequence processing capability of the self-attention network 
and to capture the features among state sequences thoroughly, 
values from a past time window for each feature are collected to form a sequence, 
which then serves as the network input. 
Moreover, as the filtering process goes, the window slides accordingly. 
Assuming the size of the time window is $s$, 
the network input at $k$ can be represented as in (\ref{InputSeq}),
\begin{equation} \label{InputSeq}
    X_{\text{input}} = \{ \Delta x_{k-s}, \cdots, \Delta x_{k-1}, \Delta y_{k-s+1}, \cdots \Delta y_{k} \}.
\end{equation}

If the historical data is insufficient (i.e., when $k < s$), 
the network fills missing feature values in the input sequence with zeros, as shown in Fig. \ref{OverallArch}(a). 
Then the network uses the processed input sequence to predict the Kalman gain $K_{k}$, 
which is subsequently utilized to update the system state.

Similar to the traditional KF, the overall filtering framework can also be summarized into prediction 
and update steps:
\begin{itemize}
   \item Prediction step:
   \begin{subequations}
           \begin{align}
               \check{x}_{k} &= f(\hat{x}_{k-1}), \\
               \hat{y}_{k}   &= h(\check{x}_{k}).
           \end{align}
   \end{subequations}

   \item Update step:
   \begin{subequations}
           \begin{align}
               K_{k} &= \text{attention}(\Delta \hat{x}_{k-1}, \Delta y_{k}), \\
               \hat{x}_{k} &= \check{x}_{k} + K_{k}(y_{k} - \hat{y}_{k}).
           \end{align}
   \end{subequations}
\end{itemize}

\subsection{Network Structure}

As shown in Fig. \ref{OverallArch}(b), 
the entire self-attention network consists of two linear embedding layers for initial processing, 
a positional encoding layer to integrate sequence position information, 
a simplified attention layer for capturing dependencies within the sequence, 
and two fully connected layers within a multi-layer perceptron (MLP) block to enrich feature representation.
It concludes with a linear mapping layer that projects the processed features to predict the Kalman gain $K_{k}$. 
Starting with input features $\Delta x_{k-1} \in \mathbb{R}^{(B, s, m)}$ and $\Delta y_{k} \in \mathbb{R}^{(B, s, n)}$,
the network fuses and transforms these inputs into a sequence $X_{\text{input}} \in \mathbb{R}^{(B, 2s, d_{\text{model}})}$ via the embedding layers.
$X$ is then enhanced for dependency recognition in the attention layer 
and further processed by the MLP to capture non-linear characteristics, \
finally outputting the predicted Kalman gain $K_{k} \in \mathbb{R}^{(B, m*n)}$.

\subsection{Network Training}

\subsubsection{Training}

We conduct end-to-end training of the entire AtKF framework. 
The training dataset is constructed from a nonlinear system (\ref{nonlinearSys}) and includes noise that is generated randomly.
Each instance in the training dataset captures the true state and system output across a defined period,
with the dataset encompassing $N$ instances, each spanning $L$ time steps. 
The $i$-th training data instance is specified as (\ref{trainData}),
\begin{equation} \label{trainData}
    D_{i} = \{(x_{j}, y_{j}) | j = 1, \cdots, L \}.
\end{equation}
Assuming $\theta$ represents all trainable parameters of AtKF, 
the network is trained using a loss function defined as (\ref{LossFunc}),
\begin{equation} \label{LossFunc}
    \mathcal{L}(\theta) = \frac{1}{L} \sum_{j=1}^{L} \Vert x_{j} - \hat{x}_{j} \Vert ^{2}.
\end{equation}
Furthermore, since the loss is derivatively related to the network output, 
the Kalman gain $K_k$, as formulated in (\ref{end2endEquation}),
\begin{equation} \label{end2endEquation}
    \begin{aligned}
       \frac{\partial \Vert x_{k} - \hat{x}_{k} \Vert ^{2}}{\partial K_{k}} 
                 &= \frac{\partial \Vert K_{k} \Delta y_{k} - \Delta x_{k} \Vert ^{2}}{\partial K_{k}} \\
                 &= 2 \cdot (K_{k} \cdot \Delta y_{k} - \Delta \tilde{x}_{k}) \cdot \Delta y_{k}^{\rm{T}},
    \end{aligned}
\end{equation}
where $\Delta \tilde{x}_{k} \triangleq x_{k} - \check{x}_{k}$, 
thus enabling the entire AtKF framework to be trained end-to-end through backpropagation.

\subsubsection{PreTraining} \label{preTrainSection}

The recursive nature of the KF leads to nested forward and backward propagation during training, 
posing risks of gradient issues and making the training unstable.
Moreover, this recursive training process fails to fully exploit the parallel processing strengths of the self-attention network.
To address these issues, we propose a pre-training method based on batch estimation and LTPWL.
Here, we first linearize the system with LTPWL 
and then use batch estimation to directly estimate the system states over a period at once,
thus generating pre-training data. 
This method avoids the recursive limitations of KF and enhances training stability and efficiency.

Batch estimation requires a linear system; 
thus, we consider approximating a nonlinear system (\ref{nonlinearSys}) through LTPWL.
Assuming that a state trajectory $X_{Tra}=\{x_{i}|i=1, \dots, L \}$ with $L$ state points is derived from 
the initial state estimate $\hat{x}_{0}$ and $f(x)$, select all state points as linearization points. 
At each point $x_{i}$, a linear segment is constructed through a first-order Taylor expansion, 
as shown in (\ref{1orderTaylor_literal}):
\begin{subequations} \label{1orderTaylor_literal}
    \begin{align}
        l_{i}^{f} &= f(x_i) + {\frac{\partial f}{\partial x}}|_{x_i}(x - x_i), \\
        l_{i}^{h} &= h(x_i) + {\frac{\partial h}{\partial x}}|_{x_i}(x - x_i),
    \end{align}
\end{subequations}
where $l_{i}^{f}$ and $l_{i}^{h}$ represent the linear segments constructed at $x_{i}$ for $f(x)$ and $h(x)$, respectively.
Then, the LTPWL expression can be constructed, with $f_{\text{LTPWL}}(x) \approx f(x)$ and $h_{\text{LTPWL}}(x) \approx h(x)$.
Assuming training data (\ref{trainData}) from time step 1 to $L$, 
the application of batch estimation for this LTPWL model is formulated as detailed in Lemma \ref{batchEstLemma}.
\begin{lemma} \label{batchEstLemma}
    Define vectors $z$ and $x$ as in (\ref{xz_definition}),
    \begin{subequations} \label{xz_definition}
        \begin{align}
            z &= \begin{bmatrix}
                \check{x}_1^{\rm{T}} & u_2^{\rm{T}} & \cdots & u_L^{\rm{T}} & \bar{y}_1^{\rm{T}} & \bar{y}_2^{\rm{T}} & \cdots & \bar{y}_L^{\rm{T}}
            \end{bmatrix}^{\rm{T}}, \label{z_definition} \\
            x &= \begin{bmatrix}
                x_1^{\rm{T}} & \cdots & x_L^{\rm{T}}
            \end{bmatrix}^{\rm{T}},
        \end{align}
    \end{subequations} 
    where $u_{k+1}=f(x_k)-{\frac{\partial f}{\partial x}}|_{x_k} x_{k}$, 
    $\bar{y}_{k} = y_{k} - (h(x_k) - {\frac{\partial h}{\partial x}}|_{x_k} x_k)$, 
    $A_{k} = {\frac{\partial f}{\partial x}}|_{x_k}$ and $C_{k}={\frac{\partial h}{\partial x}}|_{x_k}$.
    $Q_{k}$ and $R_{k}$ are noise covariance matrices in moment $k$.

    And define matrices $H$ and $W$ as in (\ref{H_matrix}) and (\ref{W_matrix}),
    \begin{equation} \label{H_matrix}
        H = \left[ \begin{array}{cccc}
            1 & & & \\
            -A_1 & 1 & & \\
            & \ddots & \ddots & \\
            & & -A_{L-1} & 1 \\
            C_1 & & & \\
            & C_2 &  & \\
            & &\ddots & \\
            & & & C_L
        \end{array} \right],
    \end{equation}
    \begin{equation} \label{W_matrix}
        W = \left[ \begin{array}{cccccccc}
            \check{P}_1 & & & & & & & \\
            & Q_{2} & & & & & & \\
            &  & \ddots & & & & & \\
            &  &  &  Q_{L} & & & & \\
            &  &  &  & R_{1} & & & \\
            &  &  &  & & R_{2} & & \\
            &  &  &  & & & \ddots & \\
            &  &  &  & & & & R_{L} \\
        \end{array} \right],
    \end{equation}
    here $\check{x}_1$ represents the prior estimate of the system state at the first moment, and $\boldsymbol{x}$ includes 
    the true states from moment 1 to $L$.
    Then the posterior estimate $\hat{x}$ satisfies (\ref{batchEst}),
    \begin{equation} \label{batchEst}
        (H^{\rm{T}} W^{-1} H)\hat{x} = H^{\rm{T}} W^{-1} z.
    \end{equation}    
    Let $\hat{x} = (H^{\rm{T}} W^{-1} H)^{-1} H^{\rm{T}} W^{-1} z$, thus completing the batch estimation over moments 1 to $L$.
\end{lemma}

\begin{proof}
    The constructed LTPWL system is equivalent to a discrete-time linear time-varying system, 
    where the dynamics of the system at each moment are determined by the state variables. 
    With the evaluation properties of LTPWL described in the preliminaries, 
    it is straightforward to derive the transition matrix $A_{k}$ 
    and observation matrix $C_{k}$ for every moment. 
    This leads to a concise system model representation, as outlined in (\ref{LTPWL_sys}),
    \begin{subequations} \label{LTPWL_sys}
        \begin{align}
            &x_{k} = \begin{matrix}
                \underbrace{{\frac{\partial f}{\partial x}}|_{x_{k-1}} x_{k-1} } \\ A_{k-1} x_{k-1}
            \end{matrix} + 
            \begin{matrix}
                \underbrace{(f(x_{k-1})-{\frac{\partial f}{\partial x}}|_{x_{k-1}} x_{k-1})} \\ u_{k}
            \end{matrix}
             + w_{k}, \\
            & \begin{matrix}
                \underbrace{y_{k} - (h(x_k) - {\frac{\partial h}{\partial x}}|_{x_k} x_k)} \\ \bar{y}_{k}
            \end{matrix}
            = 
            \begin{matrix}
                \underbrace{{\frac{\partial h}{\partial x}}|_{x_k} x_{k}} \\ C_{k} x_{k}
            \end{matrix} + v_{k},
        \end{align}
    \end{subequations}
    where $u_{k}$ and $\bar{y}_{k}$ represent the system input and observation at time $k$, $w_{k}$ and $v_{k}$ represent the 
    process and observation noise, respectively, both following Gaussian distributions. 
    According to \cite{barfoot2017state}, for the linear time-varying system (\ref{LTPWL_sys}), the batch posterior estimate 
    can be derived using (\ref{batchEst}), in which the vector $z$, the matrices $H$ and $W$ are defined as in 
    (\ref{z_definition}), (\ref{H_matrix}) and (\ref{W_matrix}), respectively.

\end{proof}

In summary, with training data (\ref{trainData}) available, 
we systematically construct the matrices $A_{i}$ and $C_{i}$ for each moment based on the true states $x_{i}$, 
facilitating batch estimation from moments 1 to $L$ as outlined in Lemma \ref{batchEstLemma}. 
Based on the batch estimated values $\hat{x} = \begin{bmatrix} \hat{x}_1^{\rm{T}} & \cdots & \hat{x}_L^{\rm{T}} \end{bmatrix}^{\rm{T}}$, 
alongside the system model $f(x)$, $h(x)$, and the initial training dataset, 
we prepare a pre-training dataset. 
This dataset comprises input features (\ref{InputSeq}) for the network and the corresponding true state values, 
with each pre-training instance specified as (\ref{preTrainData}),
\begin{equation} \label{preTrainData}
    D^{\text{pre}}_{i} = \{(X_{\text{input},j}, x_{j}) | j = 1, \cdots, L \}.
\end{equation}

The pre-training dataset is denoted as $\mathcal{D}^{\text{pre}} = \{D_{i}^{\text{pre}} | i = 1, \cdots, N \}$. 
By employing the loss function (\ref{LossFunc}), 
the network can be pre-trained in a batch, avoiding recursion. 
This pre-training serves as a better starting point for subsequent formal training.


\section{EXPERIMENTS}

This section presents simulation experiments on a two-dimensional nonlinear system \cite{revach2022kalmannet} with AtKF, 
under varying noise conditions and model mismatches. 
Results are compared to those from the traditional EKF, unscented Kalman filter (UKF), particle filter (PF) and KalmanNet.

\subsection{System Function and Parameters}

The system function is given by (\ref{syntheticNonlinearModel}), where both the system state and output are 
two-dimensional vectors, i.e., $x,y \in \mathbb{R}^2$. The parameters of the system are shown in Table 
\ref{syntheticNonlinearModel_Parameter}.
\begin{subequations} \label{syntheticNonlinearModel}
    \begin{align}
        x_{k} &= \alpha \cdot \sin (\beta \cdot x_{k-1} + \phi) + \delta + w_{k}, \label{syntheticNonlinearState} \\
        y_{k} &= a \cdot (b \cdot x_{k} + c)^2 + v_{k}. \label{syntheticNonlinearObserve}
    \end{align}
\end{subequations}
\begin{table}
    \centering
    \caption{Parameters of the Two-Dimensional Nonlinear Model}
            \begin{tabular}{cccccccc}
            \toprule
            & $\alpha$ & $\beta$ & $\phi$ & $\delta$ & $a$ & $b$ & $c$ \\
            \midrule
            ${Para}_s$ & 0.9 & 1.1 & $0.1\pi$  & 0.01 & 1 & 1 & 0 \\
            \hline
            ${Para}_m$ & 1 & 1 & 0 & 0 & 1 & 1 & 0 \\
            \bottomrule
            \end{tabular}
    \label{syntheticNonlinearModel_Parameter}
\end{table}

${Para}_s$ represents the true parameters of the system, and ${Para}_m$ denotes the parameters of the model.
The system's state transition function (\ref{syntheticNonlinearState}) and observation function 
(\ref{syntheticNonlinearObserve}) are both nonlinear. 
$w_{k}$ and $v_{k}$ represent the process noise and observation noise, respectively, both assumed to be Gaussian 
white noise with covariance matrices denoted by $Q$ and $R$.

\subsection{Experimental Setup}

We use the original nonlinear model (\ref{syntheticNonlinearModel}), starting with the initial state $x_{0} = [0.1, 0.1]^{\rm{T}}$. 
datasets for training, validation, and testing are generated under random noise. 
The training dataset contains $N=1000$ data entries, each with $L=10$ time steps. 
The validation dataset contains $N=100$ data entries, each with $L=10$ time steps. 
The test dataset contains $N=200$ data entries, each with $L=100$ time steps.
Moreover, to generate pre-training data,
a noise-free state trajectory of 10 time steps is produced using the same nonlinear model and initial state.

For training parameters, the self-attention network's sliding window size was set to $s=4$, with a batch size of 50 and 
a learning rate of 1e-4. For the AtKF, pre-training was conducted for 50 epochs and the subsequent training phase for 20 epochs. 
To ensure fairness, KalmanNet was trained for 70 epochs.
All training processes are conducted on a GTX-3090 GPU.

\subsection{Results and Analysis}

\subsubsection{Noise Robustness}

By setting the weight coefficients $q^{2} = r^{2}$ to different values, 
simulation experiments are conducted for various noise 
levels and the model used by the filter is consistent with the true system. 
The results are shown in Table \ref{diffNoise}.
\begin{table} [h]
    \centering
    \caption{Estimation Error (MSE) for the Two-Dimensional Nonlinear Model Under Different Noise Levels}
    \resizebox{0.95\linewidth}{!}{
         \begin{tabular}{cccccc}
            \toprule
            $q^{2} = r^{2}$ & 1 & 2 & 4 & 8  &  16 \\
            \midrule
            EKF        & 3.0216          & 7.6312          & 20.5524         & 64.4445      & 218.2332 \\
            \hline
            UKF        & 3.1972          & 9.8430          & 29.0027         & 103.1582      & 460.5745 \\
            \hline
            PF         & \textbf{1.4986} & \textbf{2.8381} & 5.6377         & 11.2771      & 23.9068 \\
            \hline
            KalmanNet  & 1.6303          & 3.3716          & 6.4136          & 9.6848      & 18.5984 \\
            \hline
            AtKF       & 1.6175         & 2.9235 & \textbf{4.9186} & \textbf{8.7522} & \textbf{16.6712} \\
            \bottomrule
            \end{tabular}}
    \label{diffNoise}
\end{table}
It is obvious that when the noise is significant, our AtKF shows superior performance compared with other filters.
Although our performance is similar to PF under low noise, AtKF significantly surpass PF under high noise.
Specifically when $r^{2} = 16$, AtKF achieves an MSE of 16.6712, while our performance is 30\% better than PF. 
This superiority is attributed to the attention network's ability to compensate for noise. 
Since our capability to capture dependencies among sequences and start from a good initialization point through pre-training, our results outperform KalmanNet.
Fig. \ref{stateTrajectory} shows the first component of the true state values and the estimated values from different filters 
for a selected data entry in the test dataset. It is evident that AtKF provides the best tracking of the true state. 
In conclusion, the results demonstrate that the noise characteristics in the system can be better 
fitted with the assistance of neural networks, leading to improved estimation results. Furthermore, networks based on the 
attention mechanism achieve better estimation performance by more comprehensively capturing the dependencies among sequences.

\begin{figure} 
    \centering
    \includegraphics[width=0.35\textwidth]{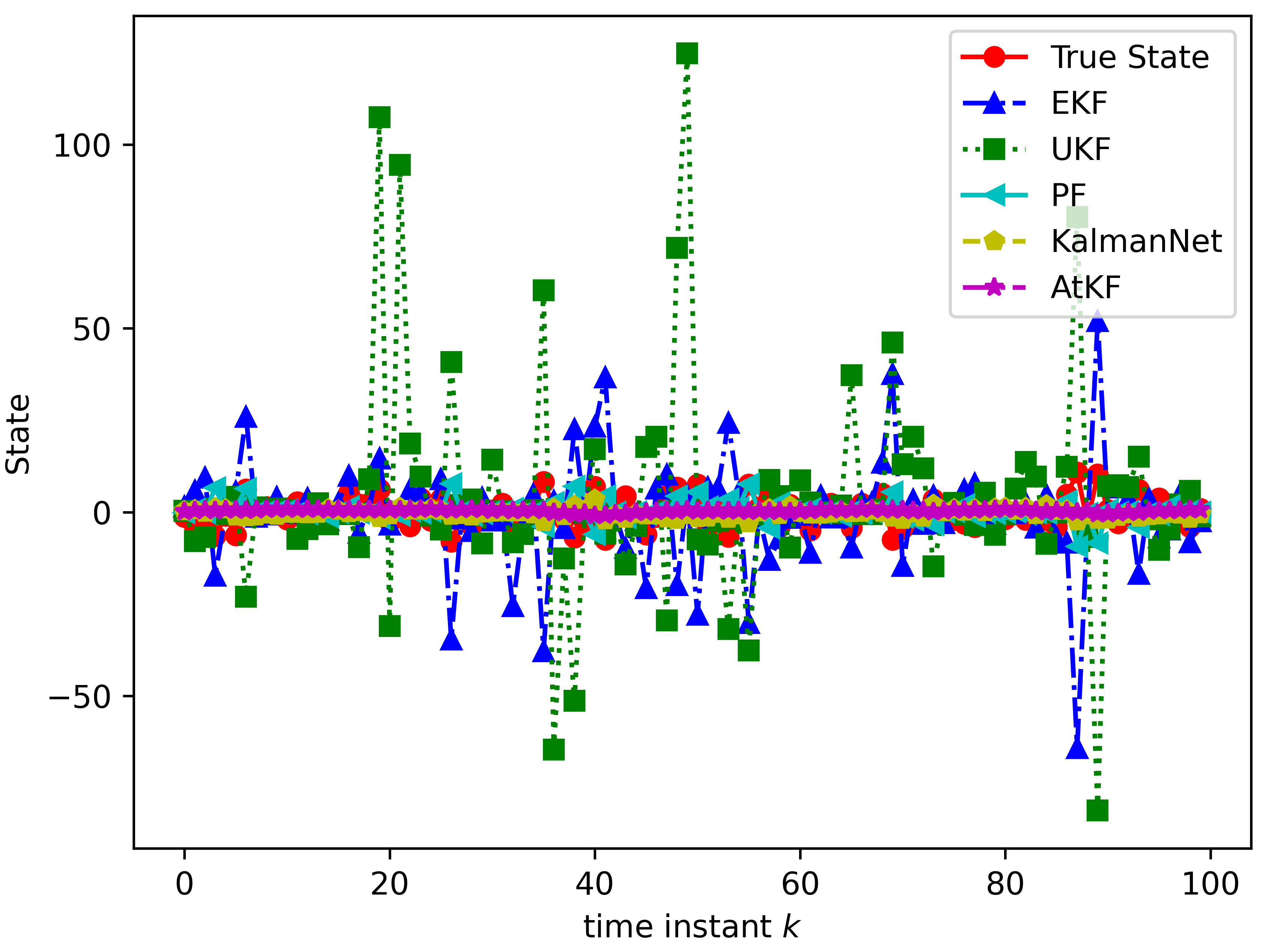}
    \caption{The true state and the estimated state (dimension 1) from different filters for data selected in the test dataset.}
    \label{stateTrajectory}
\end{figure}

\subsubsection{Model Mismatch}
Here we assume a difference between the model used by the filter and the true system. Simulation experiments are 
conducted under model mismatch conditions for different noise levels. The results are shown in Table \ref{misMatch}.
\begin{table} [h]
    \centering
    \caption{Estimation Error (MSE) for the Two-Dimensional Nonlinear Model Under Model Mismatch Conditions}
    \resizebox{\linewidth}{!}{
        \begin{tabular}{cccccc}
            \toprule
            $q^{2} = r^{2}$ & 1 & 2 & 4 & 8  &  16 \\
            \midrule
            EKF        & 3.7272          & 8.1047          & 20.2963         & 60.7735      & 211.4128 \\
            \hline
            UKF        & 3.7043          & 7.8158         & 24.1762         & 81.3889      & 319.6542 \\
            \hline
            PF         & 1.6882          & 2.9876          & 5.8008         & 11.3298      & 23.9946 \\
            \hline
            KalmanNet  & 2.0326          & 3.2508          & 6.3598          & 9.5641      & 18.6151 \\
            \hline
            AtKF       & \textbf{1.4880} & \textbf{2.8058} & \textbf{4.5026} & \textbf{8.4523} & \textbf{16.5934} \\
            \bottomrule
        \end{tabular}}
    \label{misMatch}
\end{table}
It is observed that filters embedded with neural networks remain widely better than model-based filtering. Furthermore, the 
attention mechanism network-based AtKF outperforms the GRU-based KalmanNet. This shows that AtKF offers greater robustness, 
and attention mechanism networks are better at capturing the dependencies between state 
sequences, thereby achieving superior filtering results.


\section{CONCLUSIONS}

This paper introduces the attention Kalman filter (AtKF), a novel approach that integrates self-attention with the KF to 
improve the accuracy and robustness of state estimation. AtKF addresses traditional KFs' shortcomings in handling system model 
inaccuracies and noise parameters. Specifically, AtKF uses self-attention to comprehensively capture dependencies in state 
sequences and perform an innovative pre-training strategy, which includes LTPWL for system linearization and batch estimation 
for data generation. AtKF addresses the challenges of instability and inefficiency associated with recursive training. 
Experiments with a two-dimensional nonlinear system demonstrate AtKF's effectiveness in managing noise disturbances and 
model mismatches. This work enhances KF's performance with neural network architectures and paves the way for future research 
on integrating data-driven techniques with traditional estimation methods for complex systems.

\bibliographystyle{IEEEtran}
\bibliography{IEEEexample}

\begin{thebibliography}{10}
\providecommand{\url}[1]{#1}
\csname url@rmstyle\endcsname
\providecommand{\newblock}{\relax}
\providecommand{\bibinfo}[2]{#2}
\providecommand\BIBentrySTDinterwordspacing{\spaceskip=0pt\relax}
\providecommand\BIBentryALTinterwordstretchfactor{4}
\providecommand\BIBentryALTinterwordspacing{\spaceskip=\fontdimen2\font plus
\BIBentryALTinterwordstretchfactor\fontdimen3\font minus
  \fontdimen4\font\relax}
\providecommand\BIBforeignlanguage[2]{{%
\expandafter\ifx\csname l@#1\endcsname\relax
\typeout{** WARNING: IEEEtran.bst: No hyphenation pattern has been}%
\typeout{** loaded for the language `#1'. Using the pattern for}%
\typeout{** the default language instead.}%
\else
\language=\csname l@#1\endcsname
\fi
#2}}

\bibitem{kalman1960}
R.~E. Kalman, ``{A New Approach to Linear Filtering and Prediction Problems},''
  \emph{Journal of Basic Engineering}, vol.~82, no.~1, pp. 35--45, 1960.

\bibitem{maybeck1982stochastic}
P.~S. Maybeck, \emph{Stochastic models, estimation, and control}.\hskip 1em
  plus 0.5em minus 0.4em\relax Academic press, 1982.

\bibitem{bai2023state}
Y.~Bai, B.~Yan, C.~Zhou, T.~Su, and X.~Jin, ``State of art on state estimation:
  Kalman filter driven by machine learning,'' \emph{Annual Reviews in Control},
  vol.~56, p. 100909, 2023.

\bibitem{gao2020rl}
X.~Gao, H.~Luo, B.~Ning, F.~Zhao, L.~Bao, Y.~Gong, Y.~Xiao, and J.~Jiang,
  ``Rl-akf: An adaptive kalman filter navigation algorithm based on
  reinforcement learning for ground vehicles,'' \emph{Remote Sensing}, vol.~12,
  no.~11, p. 1704, 2020.

\bibitem{tian2021state}
J.~Tian, R.~Xiong, W.~Shen, and J.~Lu, ``State-of-charge estimation of lifepo4
  batteries in electric vehicles: A deep-learning enabled approach,''
  \emph{Applied Energy}, vol. 291, p. 116812, 2021.

\bibitem{jung2020mnemonic}
S.~Jung, I.~Schlangen, and A.~Charlish, ``A mnemonic kalman filter for
  non-linear systems with extensive temporal dependencies,'' \emph{IEEE Signal
  Processing Letters}, vol.~27, pp. 1005--1009, 2020.

\bibitem{LongShortTermMemory}
S.~Hochreiter and J.~Schmidhuber, ``Long short-term memory,'' \emph{Neural
  Computation}, vol.~9, no.~8, pp. 1735--1780, 1997.

\bibitem{revach2022kalmannet}
G.~Revach, N.~Shlezinger, X.~Ni, A.~L. Escoriza, R.~J. Van~Sloun, and Y.~C.
  Eldar, ``Kalmannet: Neural network aided kalman filtering for partially known
  dynamics,'' \emph{IEEE Transactions on Signal Processing}, vol.~70, pp.
  1532--1547, 2022.

\bibitem{chung2014empirical}
J.~Chung, C.~Gulcehre, K.~Cho, and Y.~Bengio, ``Empirical evaluation of gated
  recurrent neural networks on sequence modeling,'' in \emph{NIPS 2014 Workshop
  on Deep Learning}, 2014.

\bibitem{vaswani2017attention}
A.~Vaswani, N.~Shazeer, N.~Parmar, J.~Uszkoreit, L.~Jones, A.~N. Gomez,
  {\L}.~Kaiser, and I.~Polosukhin, ``Attention is all you need,''
  \emph{Advances in Neural Information Processing Systems}, vol.~30, 2017.

\bibitem{tarela1999region}
J.~Tarela and M.~Martínez, ``Region configurations for realizability of
  lattice piecewise-linear models,'' \emph{Mathematical and Computer
  Modelling}, vol.~30, no.~11, pp. 17--27, 1999.

\bibitem{lin1994explicit}
J.-N. Lin and R.~Unbehauen, ``Explicit piecewise-linear models,'' \emph{IEEE
  Transactions on Circuits and Systems I: Fundamental Theory and Applications},
  vol.~41, no.~12, pp. 931--933, 1994.

\bibitem{tarela1990representation}
J.~Tarela, E.~Alonso, and M.~Martínez, ``A representation method for pwl
  functions oriented to parallel processing,'' \emph{Mathematical and Computer
  Modelling}, vol.~13, no.~10, pp. 75--83, 1990.

\bibitem{ovchinnikov2002max}
S.~Ovchinnikov, ``Max-min representation of piecewise linear functions,''
  \emph{Beitr{\"a}ge zur Algebra und Geometrie}, vol.~43, no.~1, pp. 297--302,
  2002.

\bibitem{xu2016irredundant}
J.~Xu, T.~J. {van den Boom}, B.~{De Schutter}, and S.~Wang, ``Irredundant
  lattice representations of continuous piecewise affine functions,''
  \emph{Automatica}, vol.~70, pp. 109--120, 2016.

\bibitem{wang2021lattice}
J.~Wang, J.~Xu, and S.~Wang, ``Lattice trajectory piecewise linear method for
  the simulation of diode circuits,'' \emph{IEEE Transactions on Circuits and
  Systems I: Regular Papers}, vol.~68, no.~5, pp. 2069--2081, 2021.

\bibitem{barfoot2017state}
T.~D. Barfoot, \emph{State estimation for robotics}.\hskip 1em plus 0.5em minus
  0.4em\relax Cambridge University Press, 2017.

\end{thebibliography}

\end{document}